\documentclass[letterpaper,10 pt, conference]{ieeeconf}
\IEEEoverridecommandlockouts                         
\usepackage{fancyhdr}
\usepackage{hyperref}
\usepackage{xspace}
\usepackage[font={small}]{caption}
\usepackage{color}
\usepackage{setspace}
\usepackage{enumerate}
\usepackage{graphics}
\usepackage{cite}
\usepackage{dsfont}
\usepackage{latexsym}
\usepackage{amsmath}
\usepackage{amssymb}
\usepackage{amsfonts}
\usepackage{amsthm}
\usepackage{times}
\usepackage{color}
\usepackage{verbatim}
\usepackage{xcolor}
\usepackage{pgfplots}
\usepackage{epsfig} 
\usepackage{subcaption}
\usepackage{hyperref}
\usepackage{algorithm}
\usepackage{algpseudocode}
\usepackage{apptools}
\AtAppendix{\counterwithin{lemma}{subsection}}
\let\labelindent\relax
\usepackage{enumitem}
\newlist{inparaenum}{enumerate}{2}
\setlist[inparaenum]{nosep}
\setlist[inparaenum,1]{label=\bfseries\arabic*.}
\setlist[inparaenum,2]{label=\arabic{inparaenumi}\emph{\alph*})}

\def\bs{\boldsymbol}

\def\mcal{\mathcal}

\def\bs{\boldsymbol}

\def\mcal{\mathcal}
\def\mcc{\mathcal{C}}
\def\mcx{\mathcal{X}}
\def\mcy{\mathcal{Y}}

\def\mcs{\mathcal{S}}
\def\mct{\mathcal{T}}

\def\bv{{\boldsymbol{v}}}
\def\balph{{\boldsymbol{\alpha}}}

\def\E{\mathbb{E}}

\def\F{\mathbb{F}}

\def\Rp{\mathbb{R}_{\geq 0}}

\def\GL{\text{GL}} 
\def\RL{\text{RL}} 

\def\UB{\text{UB}}
\def\LB{\text{LB}}

\def\indicator{\mathds{1}}

\def\numbfs{C} 
\def\BFs{\mathcal{C}} 
\def\bf{c} 
\def\overbfs{{\bf\in\BFs}}
\def\xsim{{\boldsymbol{x}\sim F_\mcx}}
\def\ysim{{\boldsymbol{y}\sim F_\mcy}}
\def\dxbf{\delta_{\bf}} 
\def\dy{\delta_\mcy} 
\def\abf{\alpha_\bf} 

\def\be{\begin{equation}}
\def\ee{\end{equation}}
\def\ba{\begin{aligned}}
\def\ea{\end{aligned}}
\def\ben{\begin{enumerate}}
\def\een{\end{enumerate}}
\def\bi{\begin{itemize}}
\def\ei{\end{itemize}}

\def\v2{\vspace{2mm}}

\theoremstyle{plain}
\newtheorem{theorem}{Theorem}[section]

\newtheorem{lemma}{Lemma}[section]

\newtheorem{definition}{Definition}

\def\bs{\boldsymbol}

\normalsize\title{\LARGE Beyond Arbitrary Allocations: \\ Security Values in Constrained General Lotto Games}

\author{
	Keith Paarporn and Jason R. Marden
	\thanks{K. Paarporn is with the Department of Computer Science, University of Colorado, Colorado Springs. 
	J. R. Marden is with the Department of Electrical and Computer Engineering, University of California, Santa Barbara. 
	Contact: \texttt{kpaarpor@uccs.edu}, \texttt{jrmarden@ece.ucsb.edu}. This work is supported in part by NSF grant \#ECCS-2346791.}
}

\begin{document}
\thispagestyle{plain}
\pagestyle{plain}
\pagestyle{empty}
\maketitle

\begin{abstract}
	Resource allocation problems across multiple contests are ubiquitous in adversarial settings, from military operations to market competition. While Colonel Blotto and General Lotto games have provided valuable theoretical foundations for such problems, their equilibrium characterizations typically permit resources to be arbitrarily allocated across all contests -- a flexibility that rarely aligns with practical constraints. This paper introduces a novel constrained variant of the General Lotto game where one player is restricted to allocating resources to only a single contest. In this model we provide lower and upper bounds on the security values for this constrained player, quantifying how the inability to distribute resources across multiple contests fundamentally changes optimal strategic behavior and performance guarantees. These findings contribute to a broader understanding of how operational constraints shape strategic outcomes in competitive resource allocation, with implications for decision-makers facing similar constraints in practice.
\end{abstract}

\section{Introduction}

Resource allocation problems are fundamental across numerous domains including control systems, operations research, and economic competition. 
These problems are characterized by three key features: limited resources that must be distributed optimally, adversarial behavior where allocation decisions are made in competition with opponents, and the need for optimization to maximize objectives under constraints. 
The strategic nature of these decisions manifests in diverse real-world scenarios -- from firms allocating capital across market sectors to maximize competitive advantage, to defense agencies distributing security resources to protect critical infrastructure against intelligent adversaries. 
In each case, decision-makers must identify \emph{admissible} allocation strategies, i.e., those that satisfy budget constraints, physical limitations, and operational requirements, to optimize the objective at hand. 

The study of competitive resource allocation has a rich history, with Colonel Blotto and General Lotto games serving as foundational mathematical frameworks for analyzing these strategic interactions \cite{Gross_1950,Roberson_2006,Schwartz_2014,kovenock2021generalizations}. 
These game-theoretic models have served as benchmark problems that have stimulated extensive research, enabling equilibrium characterizations across diverse environments, e.g., from symmetric to asymmetric contests \cite{Vu_EC2021,kovenock2021generalizations,paarporn2024reinforcement}, from complete to incomplete information structures \cite{paarporn2024incomplete,diaz2025value}, and from independent to interdependent contest structures \cite{Shahrivar_2014,Guan_2019,kovenock2021generalizations,aghajan2026defense}. 
The insights derived from these equilibrium analyses have identified salient properties of optimal resource allocation strategies, providing valuable guidance to practitioners in fields where strategic resource allocation is critical.

Despite these theoretical advances, existing equilibrium results face practical limitations that restrict their applicability. 
Typically, these equilibrium characterizations exhibit behavior where resources are arbitrarily allocated over contests, thereby potentially distributing resources in infinitesimally small amounts across all contests. 
Such flexibility rarely reflects real-world constraints. 
In defensive operations, geographic distance between battlefronts makes simultaneous deployment impractical \cite{shishika2022dynamic}, and attacks are often planned for single salient targets \cite{chowdhury2021focality}; 
in corporate settings, entering multiple markets simultaneously dilutes focus and expertise \cite{maljkovic2024blotto}; 
in cybersecurity, defender resources often cannot be fractionally distributed across all potential vulnerabilities \cite{linkov2019fundamental,iliaev2023tullock}.
These practical realities necessitate more localized deployments with a concentration in specific contests rather than resources spread thinly across many fronts. 
This disconnect between theory and practice raises a critical question: \emph{how do equilibrium conditions extend to environments where allocation decisions are structurally constrained?}

To address this gap, we introduce a novel model of resource allocation with constrained allocation decisions. 
We formulate a variant of the General Lotto game in which one player is restricted to allocating resources to only a single contest, fundamentally changing their strategic problem from distribution across multiple contests to selecting which single contest to target and how to randomize their allocation within that contest. 
Our main technical contribution, given in Theorem~\ref{thm:LBUB}, provides both lower and upper bounds on the security value of this constrained player. 
The gap between this security value and those in the classic General Lotto game showcases the impact of localization constraints in adversarial resource allocation. 
To bridge these extreme cases and provide a more complete understanding, we also conduct numerical simulations exploring the intermediate cases where players can allocate to exactly 2 contests, 
revealing how security values converge toward the unconstrained equilibrium as allocation flexibility increases.

The central focus of this paper is understanding how admissibility constraints impact equilibrium behavior and performance guarantees in adversarial resource allocation. 
While we primarily examine single-contest allocation constraints, the fundamental questions we address -- how constraints affect security values and strategic behavior -- extend naturally to other important models. 
A particularly relevant extension involves incorporating minimum investment thresholds, where players allocating to any contest must commit at least some minimum amount of resources. 
Our model can be interpreted as a special case where this minimum threshold exceeds half the available resources, effectively forcing allocation to at most one contest. 
Characterizing equilibria for these intermediate constraint models would bridge the gap between fully flexible and highly constrained allocation paradigms, advancing both theoretical understanding and practical guidance for strategic resource deployment.

\section{Problem formulation}

In this section, we first detail the classic General Lotto game in which players' allocation strategies are unrestricted.
We then propose our model formulation for a General Lotto game where one of the player's allocation strategies are restricted.

\subsection{Classic General Lotto games}

Two players $\mcx$ and $\mcy$ compete over a collection of valuable contests, $\BFs := \{1,\ldots,\numbfs\}$.
The contests have associated positive valuations $v_i > 0$ for each $\bf \in \BFs$.
We denote the vector of valuations as $\bs{v} := [v_1,\ldots,v_\numbfs]^\top$, 
and furthermore will assume a normalized total value, $\sum_{\overbfs} v_\bf = 1$.

A resource allocation for player $\mcx$ is a vector $\bs{x} = [x_1,\ldots,x_\numbfs]^\top \in \Rp^\numbfs$,
and a resource allocation for player $\mcx$ is a vector $\bs{y} = [y_1,\ldots,y_\numbfs]^\top \in \Rp^\numbfs$.
Given resource allocations $\bs{x},\bs{y}$, the payoff to player $\mcx$ is defined as
\begin{equation}\label{eq:Xpayoff_pure}
	\pi_\mcx(\bs{x},\bs{y}) := \sum_{\overbfs} v_\bf \cdot \mathds{1}\{ x_\bf \geq y_\bf \}
\end{equation}
where $\indicator\{ \cdot \}$ is the indicator function that evaluates to 1 if the statement is true, and 0 otherwise.
Under the payoff function~\eqref{eq:Xpayoff_pure}, player $\mcx$ defeats player $\mcy$ on the contests in which it allocates more resources, with ties being awarded to player $\mcx$. 
Subsequently, the payoff to player $\mcy$ is given by the sum of contest valuations that player $\mcx$ has not secured,
$\pi_\mcy(\bs{x},\bs{y}) := 1 - \pi_\mcx(\bs{x},\bs{y})$.

In the General Lotto game, the players $\mcx$ and $\mcy$ are able to randomize their allocations in any way as long as their total expected allocations do not exceed their fixed resource budgets, $X>0$ and $Y>0$, respectively.
Formally, an admissible strategy for player $\mcx$ is any cumulative distribution function (CDF) $F_\mcx$ over $\Rp^\numbfs$ that belongs to the set of distributions
\begin{equation}\label{eq:lotto_unrestricted_constraint}
	\F(X) := \left\{ F_\mcx : \text{supp}(F_\mcx) \subseteq \Rp^\numbfs, \ \E_{\bs{x}\sim F_\mcx}\left[\sum_{\overbfs} x_\bf \right] \leq X \right\}
\end{equation}
and similarly, player $\mcy$ can select any $F_\mcy \in \F(Y)$.
Given admissible strategies $(F_\mcx,F_\mcy) \in \F(X)\times\F(Y)$, we will refer to the expected payoff to player $\mcx$ (with slight abuse of notation) as
\begin{equation}\label{eq:Xpayoff}
	\pi_\mcx(F_\mcx,F_\mcy) = \E_{\substack{\bs{x} \sim F_\mcx \\ \bs{y} \sim F_\mcy}}\left[ \pi_\mcx(\bs{x},\bs{y}) \right],
\end{equation}
and subsequently, the expected payoff to player $\mcy$ is $\pi_\mcy(F_\mcx,F_\mcy) = 1 - \pi_\mcx(F_\mcx,F_\mcy)$.
The admissible set of strategies and the expected payoffs define a two-player simultaneous-move game called the \emph{General Lotto game}.
We refer to a given instance of this game as $\GL(X,Y;\bs{v})$. 

We note that $\GL(X,Y;\bs{v})$ is a two-player constant-sum game. 
Player $\mcy$'s objective is thus equivalent to \emph{minimizing} the payoff to player $\mcx$.

There are two quantities of interest regarding the performance of player $\mcx$: the \emph{max-min} and \emph{min-max} values,
\begin{equation}\label{eq:maxmin_minmax}
	\begin{aligned}
		&\max_{F_\mcx\in\F(X)} \min_{F_\mcy\in\F(Y)} \pi_\mcx(F_\mcx,F_\mcy) \\
		&\min_{F_\mcy\in\F(Y)} \max_{F_\mcx\in\F(X)} \pi_\mcx(F_\mcx,F_\mcy).
	\end{aligned}
\end{equation}
The max-min value (first line above) is the best payoff that player $\mcx$ can {guarantee}, regardless of the strategy of player $\mcy$.
The min-max value (second line above) is the lowest payoff for player $\mcx$ that player $\mcy$ can guarantee, regardless of the strategy of player $\mcx$.
The max-min inequality  states that
\begin{equation}\label{eq:SV_inequality}
	\begin{aligned}
		&\max_{F_\mcx\in\F(X)} \min_{F_\mcy\in\F(Y)} \pi_\mcx(F_\mcx,F_\mcy) \\
		&\quad\quad \leq \min_{F_\mcy\in\F(Y)} \max_{F_\mcx\in\F(X)} \pi_\mcx(F_\mcx,F_\mcy).
	\end{aligned}
\end{equation}
If it is the case that they are equivalent, then their common value is referred to as the \emph{equilibrium value}. The well-known result below characterizes the equilibrium value of $\GL(X,Y;\bs{v})$.

\begin{theorem}[Adapted from \cite{Hart_2008}]\label{thm:GL}
	Consider an instance of $\GL(X,Y;\bs{v})$. 
	Then the max-min and min-max values are equivalent.
	The equilibrium value, which we denote $\pi_\mcx^\GL$, is given by
	\begin{equation}\label{eq:GL_equil} 
    	\pi_\mcx^\GL(X,Y) := 
        \begin{cases}
            \frac{X}{2Y}, &\text{if } X \leq Y \\
			1 - \frac{Y}{2X}, &\text{if } X > Y
        \end{cases}.
    \end{equation}
\end{theorem}
We note that in $\GL(X,Y;\bv)$, the equilibrium value depends only on the ratio of the players' budgets and the total value of all contests $\sum_{\overbfs} v_\bf$. Here, since $\bs{v}$ is normalized, the total value is just 1.
Any profile $(F_\mcx^*,F_\mcy^*) \in \F(X)\times\F(Y)$ for which $F_\mcx^*$ solves the max-min problem and $F_\mcy^*$ solves the min-max problem constitutes a Nash equilibrium of the game.
That is, it satisfies the condition $\pi_\mcx(F_\mcx,F_\mcy^*) \leq \pi_\mcx(F_\mcx^*,F_\mcy^*) \leq \pi_\mcx(F_\mcx^*,F_\mcy)$
for any $F_\mcx \in \F(X)$ and $F_\mcy \in F(Y)$.

\subsection{General Lotto Game With Restricted Allocations}

In this section, we propose a formulation of the General Lotto game in which player $\mcy$ is unable to allocate resources simultaneously to more than a single contest. 
That is, a resource allocation for $\mcy$ is any vector $\bs{y} = [y_1,\ldots,y_C]^\top$ that belongs to
\begin{equation}
	\mcal{R}^1 := \left\{ \bs{y} \in \Rp^\numbfs : \bs{y} = t\cdot\bs{e}_\bf \text{ for some } \bf\in\BFs, t \geq 0 \right\}
\end{equation}
where $\bs{e}_c$ is the unit Euclidean vector in component $\bf\in\BFs$.
In other words, player $\mcy$ can only send resources to a single contest and its allocation to all other contests are zero.
There is no such restriction for player $\mcx$ -- a resource allocation for $\mcx$ is still any vector $\bs{x} = [x_1,\ldots,x_C]^\top \in \Rp^\numbfs$.


Subsequently, an admissible strategy for player $\mcy$ is any CDF $F_\mcy$ that belongs to the set of distributions
\begin{equation}\label{eq:lotto_restricted_constraint}
	\F^1(Y) := \left\{ F_\mcy \in \F(Y) : \text{supp}(F_\mcy) \subseteq \mcal{R}^1 \right\}.
\end{equation}
An admissible strategy for player $\mcx$ is, as before, any CDF $F_\mcx \in \F(X)$.
For any profile of admissible strategies $(F_\mcx,F_\mcy) \in \F(X)\times \F^1(Y)$, the expected payoff to $\mcx$ is given as in \eqref{eq:Xpayoff}, and the expected payoff to $\mcy$ is $1 - \pi_\mcx(F_\mcx,F_\mcy)$. 
This formulation defines a two-player simultaneous-move game that we term the \emph{Restricted Lotto game}.
We refer to a given instance of this game as $\RL^1(X,Y;\bs{v})$.

Just as in the classic $\GL$ game, we are interested in investigating the max-min and min-max values for the restricted game $\RL^1(X,Y;\bs{v})$.
The main difference here is that player $\mcy$ is restricted to the strategy set $\F^1(Y)$, i.e., does not have full access to strategies in $\F(Y)$.
Therefore, we consider the max-min and min-max values to be of the form \eqref{eq:maxmin_minmax}, where the set $\F(Y)$ is replaced with $\F^1(Y)$.
In this context, a corresponding restricted version of \eqref{eq:SV_inequality} still holds, i.e., the max-min value is no greater than the min-max value.

\section{Results}

In this section, we state the main result of the paper, Theorem \ref{thm:LBUB}. 
It characterizes a lower bound on the max-min value, and an upper bound on the min-max value, in the restricted game $\RL^1(X,Y;\bs{v})$.
We offer numerical validations of our result that show these bounds are tight in most cases, leading to equilibrium characterizations.
We then conduct simulation experiments that explore generalized scenarios where $\mcy$ can allocate to more than just one contest.

\subsection{Main theoretical result}

For ease of exposition, we will assume that all contests have distinct values such that they can be ordered from highest to lowest, $v_1 > v_2 > \cdots > v_\numbfs$.
We then denote $[k] \subseteq \BFs$ as the top $k$ contests in value, where $k \in \{1,\ldots,\numbfs\}$.
All of our results in this paper can be generalized to the case where not all contests are distinct.
Our main result of the paper is stated below.

\begin{theorem}\label{thm:LBUB}
	Consider any game $\RL^1(X,Y;\bs{v})$.

	\noindent 1) (Lower bound on max-min) It holds that
	\begin{equation}\label{eq:thm_LB}
		\max_{F_\mcx\in\F(X)} \min_{F_\mcy\in\F^1(Y)} \pi_\mcx(F_\mcx,F_\mcy) \geq \max_{\balph \in \Delta} \ (1 - H(\balph)) 
	\end{equation}
	where 
	\begin{equation}\label{eq:Hc}
		\begin{aligned}
			H(\balph) &:= \max_\overbfs H_\bf(\alpha_\bf) \\
			H_\bf(\alpha_\bf) &:= 
			\begin{cases}
				v_c\left(1 - \frac{\alpha_c X}{2Y}\right), &\text{if } \alpha_c < \frac{Y}{X} \\
				v_c\frac{Y}{2\alpha_c X}, &\text{if } \alpha_c \geq \frac{Y}{X} \\
			\end{cases}
		\end{aligned}
	\end{equation}
	for each $\overbfs$, and $\Delta$ is the probability simplex on $\BFs$.

	\noindent 2) (Upper bound on min-max) It holds that
	\begin{equation}\label{eq:thm_UB}
		\min_{F_\mcy \in \F^1(Y)} \max_{F_\mcx\in\F(X)} \pi_\mcx(F_\mcx,F_\mcy) \leq \min_{k=1,2,\ldots,\numbfs} \UB(\bs{q}([k]))
	\end{equation}
	where $\UB(\bs{p}) : \Delta \rightarrow [0,1]$ is defined as
	\begin{equation}\label{eq:UB_obj}
		\small
		\begin{cases}
			1 - \frac{Y}{2X}\frac{(\bv^\top \bs{p})^2}{\max_\overbfs v_\bf p_\bf}, &\text{if } \frac{Y}{X}\bv^\top \bs{p} - \max_\overbfs v_\bf p_\bf \leq 0 \\
			1 - \bv^\top \bs{p} + \frac{X}{2Y} \max_\overbfs v_\bf p_\bf, &\text{if } \frac{Y}{X}\bv^\top \bs{p} - \max_\overbfs v_\bf p_\bf > 0\\
		\end{cases}
	\end{equation}
	and $\bs{q}(\mcs) \in \Delta$ is defined for any subset of contests $\mcs \subseteq\BFs$ as
	\begin{equation}\label{eq:UB_q_explicit}
		q_\bf(\mcs) := 
		\begin{cases}
			\frac{\prod_{i\in \mcs\setminus \bf} v_j }{ \sum_{j \in \mcs} \left( \prod_{i\in \mcs\setminus j} v_i \right)  } &\text{if } \bf \in \mcs \\
			0, &\text{if } \bf \notin \mcs
		\end{cases}
	\end{equation}
	for each $\overbfs$, where $\mcs\setminus c$ denotes all elements in $\mcs$ not including $c$.
\end{theorem}

Several remarks are in order. 
First, the max-min value inherently satisfies 
\begin{equation}
	\max_{F_\mcx\in\F(X)} \min_{F_\mcy \in \F^1(Y)}   \geq \pi_\mcx^\GL(X,Y), 
\end{equation}
that is, the best guaranteed payoff for $\mcx$ in $\RL^1(X,Y;\bs{v})$ is no worse than the equilibrium payoff in the classic $\GL$ game.
In forthcoming numerical studies, we illustrate that our lower bound in \eqref{eq:thm_LB} is highly indicative of the performance improvement (resp., degradation) of player $\mcx$ (resp., player $\mcy$)
due to the restricted capabilities of player $\mcy$. 

Second, statements \eqref{eq:thm_LB} and \eqref{eq:thm_UB} together with (the restricted version of) \eqref{eq:SV_inequality}, immediately yield
\begin{equation}
	\begin{aligned}
		\LB^* &\leq \max_{F_\mcx\in\F(X)} \min_{F_\mcy \in \F^1(Y)} \pi_\mcx(F_\mcx,F_\mcy) \\
		&\quad\quad\quad\quad \leq \min_{F_\mcy \in \F^1(Y)} \max_{F_\mcx\in\F(X)} \pi_\mcx(F_\mcx,F_\mcy) \leq \UB^*
	\end{aligned}
\end{equation}
where we denote $\LB^*$ and $\UB^*$ as the right-hand sides of~\eqref{eq:thm_LB} and \eqref{eq:thm_UB}, respectively.
Thus, in the event that our chracterized bounds coincide, i.e., $\LB^* = \UB^*$, we have found the \emph{equilibrium value} of $\RL^1(X,Y;\bs{v})$.

Third, we note that the optimization associated with the lower bound \eqref{eq:thm_LB} is a concave maximization problem 
that can be handled using standard convex solvers. 
Our characterization of the upper bound \eqref{eq:thm_UB} is an analytically explicit characterization.
We devote Section~\ref{sec:LBUB} and the Appendix  for the proof of Theorem~\ref{thm:LBUB}.

\subsection{Numerical validation of Theorem \ref{thm:LBUB}}

We have numerically plotted our characterized bounds from Theorem~\ref{thm:LBUB} in Figure~\ref{fig:thm_sims} over a range of player $\mcy$'s budget $Y$.
We observe that the lower bound, $\LB^*=\max_{\balph \in \Delta}  (1 - H(\balph))$, is never less than the equilibrium payoff from the classic $\GL$ game, $\pi_\mcx^\GL$ \eqref{eq:GL_equil}.
The lower bound also reflects a significant \emph{guaranteed} improvement of player $\mcx$'s performance in $\RL^1(X,Y;\bs{v})$ as compared to the classic $\GL$ game.

The plot also demonstrates that the computed lower and upper bounds in Theorem~\ref{thm:LBUB} are quite close for low values of $Y$, 
and actually coincide for sufficiently large values of $Y$.
This indicates that the bounds of Theorem~\ref{thm:LBUB} are actually tight in this regime,
and that our methodology is a viable approach for computing equilibria of $\RL^1(X,Y;\bs{v})$.

\begin{figure}
	\centering
	\includegraphics[width=0.48\textwidth]{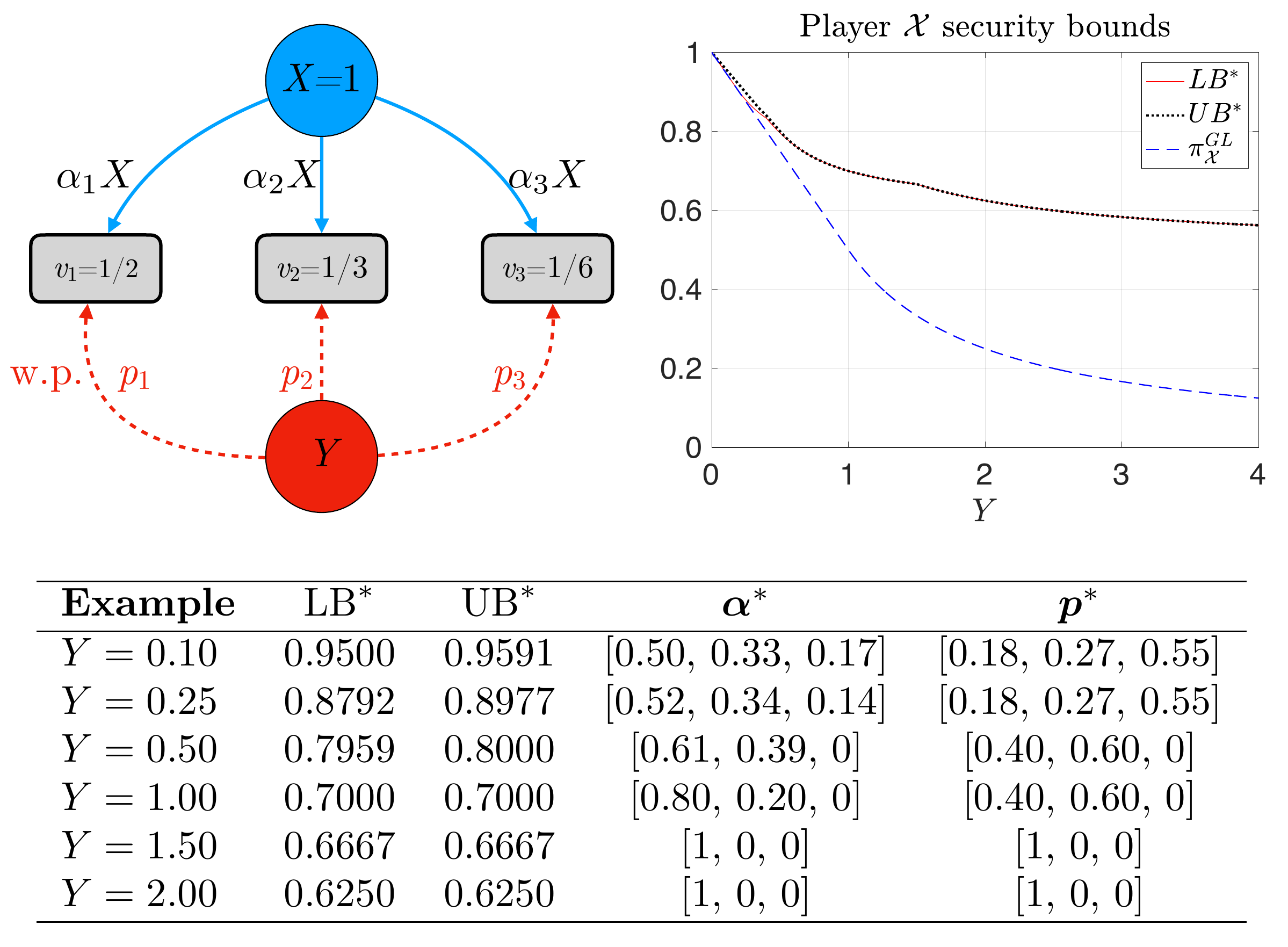}
	\caption{Example computations of lower and upper bounds on max-min and min-max values, respectively, in $\RL^1(X,Y;\bs{v})$ (Theorem~\ref{thm:LBUB}).
	We showcase a three-contest case study, with contest values shown in the top left diagram.
	We fix $X=1$ and vary $Y$.
	The top right plot shows traces for $\LB^*$ (right-hand-side of \eqref{eq:thm_LB}; solid red), $\UB^*$ (right-hand-side of \eqref{eq:thm_UB}; dotted black), and the equilibrium value from the classic $\GL$ game, $\pi_\mcx^\GL$ (dashed blue).
	We used standard convex solvers to calculate $\LB^*$, and directly applied formulas \eqref{eq:thm_UB} and \eqref{eq:GL_equil} to calculate $\UB^*$ and $\pi_\mcx^\GL$, respectively.
	We note a very small gap between $\LB^*$ and $\UB^*$ for values of $Y < 1$, and 
	they actually coincide for sufficiently high $Y \geq 1$, indicating equilibrium solutions.
	The optimizers $\balph^*$ to \eqref{eq:thm_LB} and $\bs{p}^*$ to \eqref{eq:thm_UB} for six examples are shown in the bottom table.
	We observe that player $\mcy$'s strategy $\bs{p}^*$ randomly selects among all three contests when it has a small budget, $Y \leq 0.25$.
	Interestingly, the contest selection probabilities are not proportional to the contest valuations, and it focuses its attack more on the least valuable contest, $v_3$.
	For higher budgets, player $\mcy$ abandons the lesser-valued contests in favor of the highest-value contest.}\label{fig:thm_sims}
\end{figure}

In the table of Figure~\ref{fig:thm_sims}, we refer to $\balph^*$ as the optimizer to the right-hand side of \eqref{eq:thm_LB}, and $\bs{p}^*$ as the vector among the $\bs{q}([k])$ that optimizes the right-hand side of \eqref{eq:thm_UB}.
Here, $\abf^*$ is the fraction of its total budget that player $\mcx$ devotes to contest $\bf$ (in expectation), and
$p^*_\bf$ is the probability that player $\mcy$ selects a contest $\bf$.
By \eqref{eq:thm_UB}, it selects only among the top $k^*$, where $k^*$ is determined from the game's parameters $X,Y,\bs{v}$.

\begin{figure}
	\centering
	\includegraphics[scale=0.35]{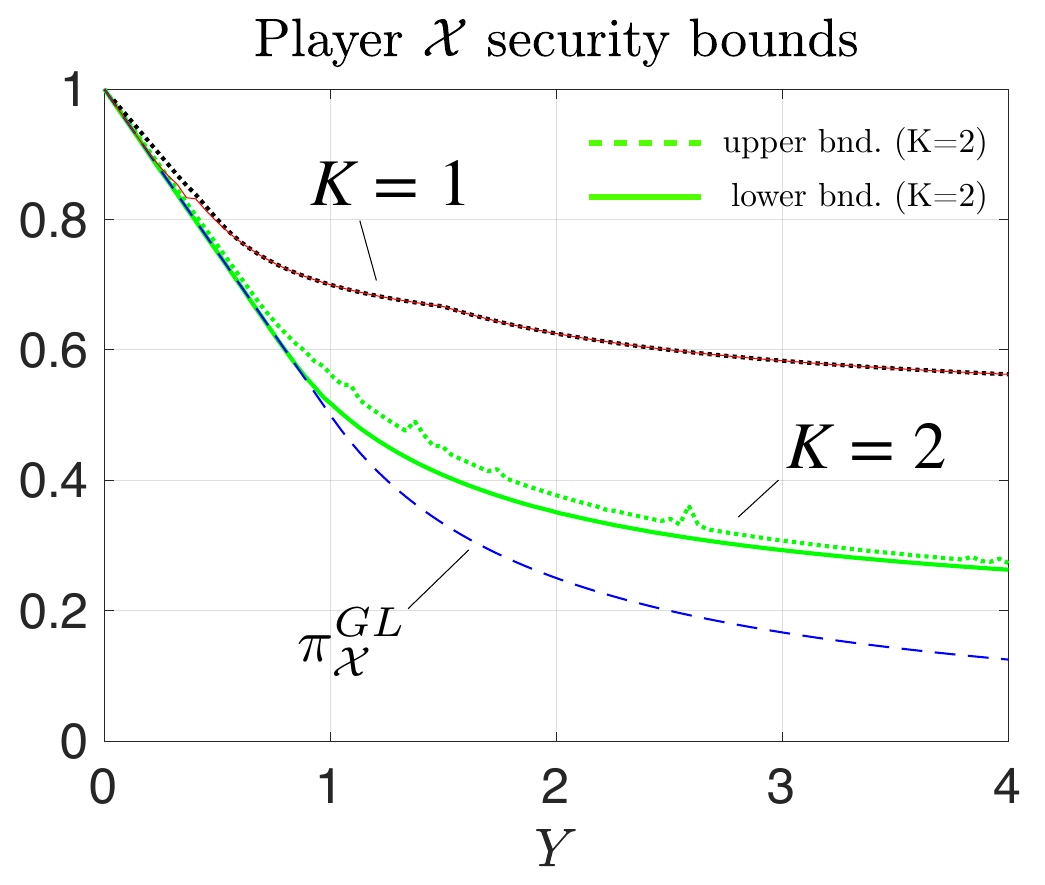}
	\caption{This plot shows lower and upper bounds on max-min and min-max values in the scenario where player $\mcy$ can allocate to $K=2$ contests.
	Here, we consider the same three-contest setup from Figure~\ref{fig:thm_sims}.
	The traces that correspond to $K=1$ and $\pi_\mcx^\GL$ are reproduced from Figure~\ref{fig:thm_sims}.
	The new green traces correspond to $K=2$: the solid green is a lower bound, and the dotted green is an upper bound.
	Both are obtained from approximate numerical solutions to the problems \eqref{eq:OPT-X-K} and \eqref{eq:OPT-Y-K}, respectively.
	We observe that the approximate lower bound (solid green) is never less than $\pi_\mcx^\GL$, the equilibrium payoff in the classic $\GL$ game where $\mcy$ is unrestricted (equivalently, the $K=3$ scenario).
	It also indicates that for larger values of $Y$, player $\mcx$ sees a significant improvement in its guaranteed performance when $K=2$ relative to $K=3$, but not as much as when $K=1$. 
	We also observe small gaps between the approximate $K=2$ lower and upper bounds.
	}\label{fig:K2}
\end{figure}

\subsection{Numerical study of generalized scenarios}\label{sec:K2}

In this subsection, we numerically explore generalized scenarios where player $\mcy$ is not as restricted, and is capable of allocating to exactly $K > 1$ contests simultaneously.
We refer these scenarios as $\RL^K(X,Y;\bs{v})$.
To do so, we formulated two optimization problems analogous to \eqref{eq:thm_LB} and \eqref{eq:thm_UB} that determine lower and upper bounds on the security values associated with $\RL^K(X,Y;\bs{v})$.
We omit their derivations due to space limitations. However, they follow similar steps to the proofs of \eqref{eq:thm_LB} and \eqref{eq:thm_UB} (detailed in Section \ref{sec:LBUB}).

The lower bound problem is stated as follows:
\begin{equation}\label{eq:OPT-X-K}
	\max_{\balph\in\Delta,\delta\in[0,1]} \left\{1 - \max_{\mcs \in \BFs^K} H_\mcs(\balph,\delta)\right\}
\end{equation}
where $\BFs^K$ is the collection of subsets of  size $K$, and we define for any $\mcs \in \BFs^K$,
\begin{equation}
	H_\mcs(\balph,\delta) := v_\mcs (1-\delta) + \delta^2 \frac{Y}{2X}\max_{\bf\in\mcs} \left\{ \frac{v_\bf}{\abf} \right\}
\end{equation}
and $v_\mcs = \sum_{\bf\in\mcs} v_\bf$.
We note that the objective function in \eqref{eq:OPT-X-K} serves as a lower bound for the max-min value of player $\mcx$ in $\RL^K(X,Y;\bs{v})$ for any choice of $(\balph,\delta)$.

The upper bound problem is stated as follows:
\begin{equation}\label{eq:OPT-Y-K}
	\min_{\bs{p},\bs{\beta}} \ \ 
	\begin{cases}
		1 - \frac{T_1^2(\bs{p})}{4T_2(\bs{p},\bs{\beta})}, &\text{if } T_1(\bs{p}) \leq 2T_2(\bs{p},\bs{\beta}) \\
		1 - T_1(\bs{p}) + T_2(\bs{p},\bs{\beta}), &\text{if } T_1(\bs{p}) > 2T_2(\bs{p},\bs{\beta})
	\end{cases}
\end{equation}
where we define
\begin{equation}
	\begin{aligned}
		T_1(\bs{p}) &:= \left(\sum_\overbfs v_\bf \sum_{\mcs\in\BFs^K : \bf\in\mcs} p_\mcs \right) \\
		T_2(\bs{p},\bs{\beta}) &:= \frac{X}{2Y} \max_\overbfs \left\{ v_\bf \sum_{\mcs\in\BFs^K : \bf\in\mcs} \frac{p_\mcs}{\beta_{\mcs,\bf}} \right\}.
	\end{aligned}
\end{equation}
The decision variables $\bs{p} = \{p_\mcs\}_{\mcs\in\BFs^K}$, where $\BFs^K$ are all subsets of contests of size $K$, are subject to constraints $p_\mcs \geq 0$ for all $\mcs\in\BFs^K$ and $\sum_{\mcs\in\BFs^K} p_\mcs = 1$. 
The decision variables $\bs{\beta} = \{\beta_{\mcs,\bf} \}_{\mcs\in\BFs^K, c\in\mcc}$ are subject to constraints
$\beta_{\mcs,\bf} \geq 0$ for all $\bf \in \mcs$, $\beta_{\mcs,\bf} = 0$ if $\bf\notin\mcs$, and $\sum_{\bf\in\mcs} \beta_{\mcs,\bf} = 1$.
We note that the objective function in \eqref{eq:OPT-Y-K} serves as an upper bound for the min-max value of player $\mcx$ in $\RL^K(X,Y;\bs{v})$ for any choice of $(\bs{p},\bs{\beta})$.

With the two optimization problems \eqref{eq:OPT-X-K} and \eqref{eq:OPT-Y-K} in hand, we apply off-the-shelf solvers (\texttt{fmincon} in MATLAB) to obtain approximate solutions.
It holds that approximate (i.e., sub-optimal) solutions to the \eqref{eq:OPT-X-K} and \eqref{eq:OPT-Y-K} still serve as lower and upper bounds to the max-min and min-max values in $\RL^K(X,Y;\bs{v})$, respectively.

In Figure~\ref{fig:K2}, we considered $K=2$ for the same three-contest setup from Figure~\ref{fig:thm_sims}.
Figure~\ref{fig:K2} reproduces the plot of Figure~\ref{fig:thm_sims}, along with new traces for lower and upper bounds of $\RL^2(X,Y;\bs{v})$.
We find that the computed lower bound for $K=2$ (solid green) is never less than $\pi_\mcx^\GL$, and
is also never greater than the lower bound $\LB^*$ associated with $K=1$.
Although they never coincide in this plot, we observe that the numerical lower and upper bounds for $K=2$ are quite close.


\section{Establishing the upper and lower bounds}\label{sec:LBUB}

In this section, we develop a proof of Theorem~\ref{thm:LBUB} by establishing the lower and upper bounds \eqref{eq:thm_LB} and \eqref{eq:thm_UB}.


\subsection{Establishing the lower bound}

In our analysis, we will consider strategies for player $\mcx$ of the form defined  below.
\begin{definition}\label{def:Xclass}
	Consider any game $\RL^1(X,Y;\bs{v})$.
	We define $\hat\F(X) \subset \F(X)$ as the set of CDFs $F_\mcx$ whose univariate marginal CDFs can be written in the following form:
	For $\overbfs$ and for any $u_\bf \geq 0$,
	\begin{equation}\label{eq:Xclass}
		F_{\mcx,\bf}(u_\bf) = 1 - \dxbf + \dxbf\min\left\{\frac{\dxbf}{2\abf X}u_\bf, 1 \right\}
	\end{equation}
	where the parameters satisfy $\dxbf \in [0,1]$, $\abf \geq 0$ with $\sum_\overbfs \abf = 1$.
	We denote $\bs{\delta} = [\delta_1,\ldots,\delta_\numbfs]^\top$ and $\balph = [\alpha_1,\ldots,\alpha_\numbfs]^\top$.
\end{definition}
For any strategy $F_\mcx \in \hat\F(X)$, its marginal distributions on allocations to each contest is a uniform distribution with a point mass at zero.
The weight of the point mass is $1-\dxbf$, which is interpreted as the probability that $\mcx$ sends zero resources to $\bf$.
The parameters $\abf$ are the fractions of $\mcx$'s total resources that it sends to each contest $\bf$ in expectation.
They control the length of the marginal uniform distributions such that $F_\mcx$ is admissible.
We re-emphasize that a strategy $F_\mcx \in \hat\F(X)$ is \emph{not} restricted: it can randomize over full allocations $\bs{x}$ for which $x_\bf > 0$ for all $\overbfs$. Any strategy $F_\mcx \in \hat\F(X)$ allocates exactly $X$ resources in expectation, and therefore $\hat\F(X) \subset \hat\F(X)$.

By considering strategies that belong to $\hat\F(X)$, we can pose a finite-dimensional optimization problem associated with deriving a lower bound for the max-min value.
The following lemma establishes the lower bound in \eqref{eq:thm_LB}.


\begin{lemma}\label{lem:LB}
	Consider any game $\RL^1(X,Y;\bs{v})$.
	Then
	\begin{equation}\label{eq:OPT-X}
		\max_{F_\mcx\in\F(X)} \min_{F_\mcy\in\F^1(Y)} \pi_\mcx(F_\mcx,F_\mcy) \geq \max_{\balph\in\Delta} \ (1 - H(\balph)).
	\end{equation}
	where $H(\cdot)$ is defined as in \eqref{eq:Hc}.
\end{lemma}
\begin{proof}
	Consider any $(F_\mcx,F_\mcy) \in \hat\F(X)\times \F^1(Y)$. 
	The payoff to player $\mcx$ can be written
	\begin{equation}\label{eq:LB_intermediate1}
		\begin{aligned}
			&\pi_\mcx(F_\mcx,F_\mcy) = 1 - \pi_\mcy(F_\mcx,F_\mcy) \\ 
			&=1 - \E_{\ysim}\left[ \E_{\xsim} \left[ \sum_\overbfs v_\bf\cdot \indicator\{x_\bf < y_\bf \} \right]  \right].
		\end{aligned}
	\end{equation}
	Because ties favor player $\mcx$, the event $(x_\bf < y_\bf)$ can only happen when $\mcy$ has allocated a positive amount of resources to contest $\bf$.
	Let us denote $P_\bf := \E_{\bs{y}\sim F_\mcy}[\indicator\{ y_\bf > 0 \}] \geq 0$ as the probability that $y_\bf > 0$ under the distribution $F_\mcy \in \F^1(Y)$.
	Since $y_\bf > 0$ can only be true for a single $\bf$, it holds that $\sum_\overbfs P_\bf = 1$. 
	Conditioned on the event $y_\bf > 0$, we write the conditional $\bf$-marginal distribution as $F_{\mcy|\bf}$.
	Continuing from the previous calculation, we obtain
	\begin{equation}\label{eq:LB_intermediate2}
		\begin{aligned}
			&\pi_\mcx(F_\mcx,F_\mcy) = 1 -  \sum_\overbfs v_\bf P_\bf \E_{y_c\sim F_{\mcy|\bf}}[F_{\mcx,c}(y_c)] \\
			&= 1 - \sum_\overbfs v_\bf P_\bf \E_{y_c\sim F_{\mcy|\bf}}[1-\dxbf + \dxbf\min\left\{ \frac{\dxbf}{2\abf X}y_\bf, 1 \right\}] \\
			&\geq 1 -  \sum_\overbfs v_\bf P_\bf \left[1-\dxbf + \frac{\dxbf^2}{2\abf X}\E_{y_c\sim F_{\mcy|\bf}}[y_\bf] \right] \\
			&= 1 - \sum_\overbfs v_\bf P_\bf \left[1-\dxbf + \frac{\dxbf^2}{2\abf X}Y \right].
		\end{aligned}
	\end{equation}
	The inequality is due to $\min\{a,b\} \leq a$ for any numbers $a,b$.
	The last equality follows because the expectation of $y_\bf$, conditioned on $y_\bf > 0$, must be $Y$ due to the budget constraint \eqref{eq:lotto_restricted_constraint}.
	Now, we have
	\begin{equation}\label{eq:LB_intermediate3}
		\begin{aligned}
			&\min_{F_\mcy\in\F^1(Y)} \pi_\mcx(F_\mcx,F_\mcy) \\
			&\quad\quad\geq 1 - \max_\overbfs\left\{v_\bf\left( 1-\dxbf + \frac{\dxbf^2}{2\abf X}Y \right) \right\} \\
			&\quad\quad=: T_\LB(\bs{\delta},\bs{\alpha};\bv).
		\end{aligned}
	\end{equation}
	The inequality results from minimizing the last expression of \eqref{eq:LB_intermediate2} over $F_\mcy \in \F^1(Y)$, or equivalently, over $\bs{P} \in \Delta$.
	We observe that $T_\LB(\bs{\delta},\bs{\alpha};\bv)$ can be analytically maximized over $\bs{\delta} \in [0,1]^\numbfs$.
	The optimal choices are
	\begin{equation}
		\dxbf^* := \min\left\{\frac{\abf X}{Y}, 1 \right\}, \ \overbfs.
	\end{equation}
	It holds that $T_\LB(\bs{\delta}^*,\bs{\alpha};\bv) = 1 - H(\balph)$, and we obtain the result.
\end{proof}

\subsection{Establishing the upper bound}

Here, we will consider strategies for player $\mcy$ of the form defined below.

\begin{definition}\label{def:Y1class}
	Consider any game $\RL^1(X,Y;\bs{v})$. 
	We define $\hat\F^1(Y) \subset \F^1(Y)$ as the set of CDFs $F_\mcy$ which can be written in the following form:
	For any $\bs{u} \in \Rp^\numbfs$,
	\begin{equation}
		F_\mcy(\bs{u}) = 1-\dy + \dy \sum_{\overbfs} p_\bf \min\left\{\frac{\dy}{2Y}u_\bf, 1 \right\}
	\end{equation}
	where the parameters satisfy $\dy \in [0,1]$, $p_\bf \geq 0$ for all $\overbfs$, and $\sum_{\overbfs} p_\bf = 1$.
\end{definition}
A strategy $F_\mcy \in \hat\F^1(Y)$ is interpreted as follows. 
With independent probability $1-\dy$, player $\mcy$'s resource allocation is zero on all contests.
With probability $\dy$, player $\mcy$ selects a single contest $\bf$ according to the probability vector $\bs{p} := [p_1,\ldots,p_\numbfs]^\top$.
The amount of resources allocated to contest $\bf$ is $y_\bf = \frac{2Y}{\dy}\cdot U_c$, where $U_c \sim\text{Unif}[0,1]$ is a uniform random variable.
The amount allocated to any other contest is zero, i.e., $y_k = 0$ for all $k\neq \bf$. Any strategy $F_\mcy \in \hat\F^1(Y)$ allocates exactly $Y$ resources in expectation, and therefore $\hat\F^1(Y) \subset \hat\F(Y)$.

By considering strategies that belong to $\hat\F^1(Y)$, we can pose a finite-dimensional optimization problem associated with deriving an upper bound for the min-max value.
The following lemma establishes the upper bound in \eqref{eq:thm_UB}.

\begin{lemma}\label{lem:UB_opt}
	Consider any game $\RL^1(X,Y;\bs{v})$. 
	Then
	\begin{equation}
		\min_{F_\mcy \in \F^1(Y)} \max_{F_\mcx \in \F(X)} \pi_\mcx(F_\mcx,F_\mcy) \leq \min_{\bs{p}\in\Delta} \ \text{UB}(\bs{p})
	\end{equation}
	where $\text{UB}(\bs{p})$ is defined as in \eqref{eq:UB_obj}.
\end{lemma}
\begin{proof}
	Consider any $(F_\mcx,F_\mcy) \in \F(X)\times \hat\F^1(Y)$.
	The payoff to player $\mcx$ can be written as 
	\begin{equation}\label{eq:UB_intermediate1}
		\begin{aligned}
			&\pi_\mcx(F_\mcx,F_\mcy) =\E_{\xsim}\left[ \E_{\ysim} \left[ \sum_c v_c\cdot \indicator\{y_c \leq x_c \} \right]  \right] \\
			&= \sum_\overbfs v_c\cdot \E_{\xsim}\left[ F_{\mcy,\bf}(\bs{x})  \right] \\
			&= \sum_\overbfs v_\bf\cdot \E_{\xsim}\left[ (1 - p_\bf\dy) + p_\bf\dy \min\left\{\frac{\dy}{2Y}x_\bf, 1 \right\}  \right] \\
			&\leq \sum_\overbfs v_c\cdot \E_{\xsim}\left[ (1 - p_\bf\dy) + p_\bf\dy \frac{\dy}{2Y}x_\bf  \right] \\
			&= 1 - \dy + \frac{\dy^2}{2Y}\sum_\overbfs v_\bf p_\bf X_\bf \\
		\end{aligned}
	\end{equation}
	where we write $X_\bf := \E_\xsim[x_\bf]$.
	The second and third equalities above follow from writing the univariate marginal distributions $F_{\mcy,\bf}(u_\bf) := \lim_{u_i \rightarrow\infty, \forall i \neq \bf} F_\mcy(\bs{u})$, from Definition \ref{def:Y1class}.
	The inequality is due to the fact that $\min\{a,b\} \leq a$ for any numbers $a,b$.
	We then have
	\begin{equation}\label{eq:UB_intermediate2}
		\begin{aligned}
			&\max_{F_\mcx\in\F(X)} \pi_\mcx(F_\mcx,F_\mcy) \\
			&\quad\quad\leq 1 - \dy + \dy^2\frac{X}{2Y} \max_\overbfs \{v_\bf p_\bf\} \\
			&\quad\quad=: T(\dy,\bs{p}).
		\end{aligned}
	\end{equation}
	The inequality results from observing that any strategy $F_\mcx \in\F(X)$ that maximizes 
	the last expression of \eqref{eq:UB_intermediate1} is one that places all resources $X$ (in expectation) on a contest that has the maximal value among $v_\bf p_\bf$.
	
	By Definition \ref{def:Y1class}, any $F_\mcy \in \hat\F^1_\mcy(Y)$ is associated with a pair of parameters $(\dy,\bs{p})$.
	Notice that $T(\dy,\bs{p})$ can be analytically minimized over $\dy\in[0,1]$,
	\begin{equation}
		\dy^* := \min\left\{ \frac{\bv^\top \bs{p}}{\max_\overbfs v_\bf p_\bf} \frac{Y}{X}, 1 \right\}.
	\end{equation}
	We then obtain $T(\dy^*,\bs{p}) = \text{UB}(\bs{p})$ as defined in \eqref{eq:UB_obj}, and we obtain the result.
\end{proof}

The value $\min_{\bs{p}\in\Delta} \ \text{UB}(\bs{p})$ serves as an upper bound, but it is associated with a non-convex minimization problem. 
Nevertheless, it is possible to provide an explicit characterization of this value.

\begin{lemma}\label{lem:UB_analytical}
	Consider any game $\RL^1(X,Y;\bs{v})$.
	Then
	\begin{equation}
		\min_{\bs{p}\in\Delta} \ \text{UB}(\bs{p}) = \min_{k=1,2,\ldots,\numbfs} \UB(\bs{q}([k]))
	\end{equation}
	where $\bs{q}([k])$ is defined as in \eqref{eq:UB_q_explicit}, and $[k]\subseteq \BFs$ are the top $k$ contests in value.
\end{lemma}
\begin{proof}
	The detailed proof is given in the Appendix.
\end{proof}

With Lemmas \ref{lem:LB}, \ref{lem:UB_opt}, and \ref{lem:UB_analytical}, we attain a proof of Theorem \ref{thm:LBUB}.

\section{Conclusion and Future Work}

In this paper, we proposed a novel formulation of the General Lotto game where one of the players is restricted to allocating resources to only a single contest, and the other player is not restricted.
Our study seeks to understand how this asymmetry in players' capabilities impacts their strategic decision-making.
The results establish lower and upper bounds on the players' security values, i.e., their guaranteed performance. 
We find that our bounds are quite close, and even coincide in many instances, which indicates equilibrium solutions of our formulation.
These resulting performance metrics also reflect a significant impact as compared to the classic Lotto game, where both players are unrestricted.

This study can be extended by establishing conditions for when the lower and upper bounds coincide, allowing claims for equilibrium solutions.
Further study will also more deeply investigate generalized scenarios where the restricted player can allocate to exactly $K > 1$ contests.

\bibliographystyle{IEEEtran}
\bibliography{sources}

\appendix

\subsection{Proof of Lemma~\ref{lem:UB_analytical}}

Here, we provide a detailed proof for Lemma~\ref{lem:UB_analytical}, 
the analytical characterization \eqref{eq:thm_UB} of the upper bound in Theorem~\ref{thm:LBUB}.
Recall that $\UB(\bs{p})$ \eqref{eq:UB_obj} is defined as
\begin{equation}
	\small
	\begin{cases}
		1 - \frac{Y}{2X}\frac{(\bv^\top \bs{p})^2}{\max_\overbfs v_\bf p_\bf}, &\text{if } \frac{Y}{X}\bv^\top \bs{p} - \max_\overbfs v_\bf p_\bf \leq 0 \\
		1 - \bv^\top \bs{p} + \frac{X}{2Y} \max_\overbfs v_\bf p_\bf, &\text{if } \frac{Y}{X}\bv^\top \bs{p} - \max_\overbfs v_\bf p_\bf > 0\\
	\end{cases}
\end{equation}
Let us denote the optimization problem \eqref{eq:OPT-Y} as
\begin{equation}\label{eq:OPT-Y}
	\min_{\bs{p}\in\Delta} \UB(\bs{p}). \tag{OPT-Y}
\end{equation}


\begin{definition}
	For any subset of contests $\mcs \subseteq \BFs$, define
	\begin{equation}\label{eq:Qs}
		Q_\mcs := \left\{\bs{p} \in \Delta : i \in \arg\max_\overbfs v_\bf p_\bf, \ \forall i\in\mcs \right\}
	\end{equation}
	We also define 
	\begin{equation}
		\Delta_\mcs := \left\{ \bs{p}\in\Delta : p_i = 0 \ \forall i \notin\mcs \right\}.
	\end{equation}
\end{definition}
The set $Q_\mcs$ is the set of probability vectors for which $v_i p_i = v_j p_j$ for all contests $i,j\in\mcs$, 
and this value is maximal among all contests $\overbfs$. 
The set $\Delta_\mcs$ is the set of all probability vectors whose support is constrainted to the contests $\mcs$.

The above Definition sheds light on the points $\bs{q}(\mcs)$ as defined in \eqref{eq:UB_q_explicit}.
Indeed, $\bs{q}(\mcs)$ is the unique probability vector with support constrained to $\mcs$, that equalizes the values $v_\bf p_\bf$ for all $\bf\in\mcs$.
It follows that we can alternately write $\bs{q}(\mcs) = Q_\mcs \cap \Delta_\mcs$.


\begin{lemma}\label{lem:UB_conv_hull}
	Let $\mcs\subseteq\BFs$ be a subset of contests. Then
	\begin{equation}
		Q_\mcs = \text{conv}\left\{ \bs{q}(\mcal{T}) \right\}_{\mct\supseteq \mcs},
	\end{equation}
	where $\text{conv}(\cdot)$ indicates the convex hull. 
\end{lemma}
This lemma states that the set of probability vectors $Q_\mcs$ is the convex polytope
whose vertices are $\bs{q}(\mct)$, where $\mct$ ranges over all supersets of $\mcs$.
\begin{proof}
	We first show that $\text{conv}\left\{ \bs{q}(\mcal{T}) \right\}_{\mct\supseteq \mcs} \subseteq Q_\mcs$.
	Suppose $\bs{p} \in \text{conv}\left\{ \bs{q}(\mcal{T}) \right\}_{\mct\supseteq \mcs}$.
	Then there exists convex weights $\{w_\mct\}_{\mct\supseteq \mcs}$ such that 
	\begin{equation}
		\bs{p} = \sum_{\mct\supseteq \mcs} w_\mct \bs{q}(\mcal{T}).
	\end{equation}
	Every vertex $\bs{q}(\mct)$ has the property that $v_i q_i(\mct) = v_j q_j(\mct)$ for all $i,j \in \mcs$.
	Then,
	\begin{equation}
		\begin{aligned}
			v_i p_i = \sum_{\mct\supseteq \mcs} w_\mct v_i q_i(\mct) = \sum_{\mct\supseteq \mcs} w_\mct v_j q_j(\mct) = v_j p_j
		\end{aligned}.
	\end{equation}
	Now, we show $Q_\mcs \subseteq \text{conv}\left\{ \bs{q}(\mcal{T}) \right\}_{\mct\supseteq \mcs}$.
	Suppose this were not true, i.e., there exists a $\bs{p} \in Q_\mcs$ that cannot be written as a convex combination of the vectors $\{ \bs{q}(\mcal{T}) \}_{\mct\supseteq \mcs}$.
	The set $Q_\mcs$ itself is a convex polytope (convex hull of a finite number of vertices) because it can alternately be defined by a set of linear inequalities.
	Every vertex $\bs{r}$ of $Q_\mcs$ must have the property that $v_i r_i = v_j r_j$ for all $i,j \in \mcs$.
	This is to say that $\bs{r}$ must have support over $\mcs$.
	Moreover, there must exist a vertex $\bs{r}$ of $Q_\mcs$ that is not part of the collection $\{ \bs{q}(\mcal{T}) \}_{\mct\supseteq \mcs}$.
	Suppose the support of such a vertex $\bs{r}$ is $\mct \supseteq \mcs$.
	This vertex is the probability vector lying in $Q_\mcs$ with support over $\mct$, i.e., $\bs{r} = Q_\mcs \cap \Delta_\mct$.

	Since $Q_\mct \subseteq Q_\mcs$ and $\Delta_\mct \subset Q_\mct$, we can also write $\bs{r} = Q_\mct \cap \Delta_\mct$.
	However, this is precisely $\bs{q}(\mct)$. This leads to a contradiction.
\end{proof}

The next lemma provides expressions for $\UB(\bs{q}(\mcs))$ at the equalizing points $\bs{q}(\mcs)$. 

\begin{lemma}\label{lem:UB_obj_equalized}
    Consider any $\mcs \subseteq \mcc$, where $|\mcs| = k$. Then,
    \begin{equation}
        \UB(\bs{q}(\mcs)) =
        \begin{cases}
            &1-\frac{k^2}{2}\frac{Y}{X} \cdot \frac{\prod_{j\in\mcs} v_j}{\sum_{\ell \in \mcs} \left( \prod_{j\in \mcs\setminus \ell} v_j \right)}, \\
			&\hspace{40mm}\text{if }   kY \leq X \\
            &1-\left(k - \frac{X}{2Y}\right) \cdot \frac{\prod_{j\in\mcs} v_j}{\sum_{\ell \in \mcs} \left( \prod_{j\in \mcs\setminus \ell} v_j \right)}, \\
			&\hspace{40mm}\text{if } kY > X
        \end{cases}.
    \end{equation}
\end{lemma}
\begin{proof}
	This result follows directly from applying the expressions \eqref{eq:UB_q_explicit} to the objective function \eqref{eq:UB_obj}.
\end{proof}

The next lemma establishes that, when restricted to certain sub-domains, the objective function is concave.

\begin{lemma}\label{lem:UB_concave}
    For each $\overbfs$, consider the function $\UB_\bf(\bs{p})$, defined as $\UB(\bs{p})$ restricted to the domain $Q_{\{\bf\}}$. 
	Then, $\UB_c(\bs{p})$ is concave.
\end{lemma}

\begin{proof}
    To show $\UB_\bf(\bs{p})$ is concave, it suffices to show that for any $\bs{p},\bs{p}' \in Q_{\{\bf\}}$,
	the single-variable function $UB_\bf(t\bs{p} + (1-t)\bs{p}')$ for $t\in[0,1]$ is concave.

	From \eqref{eq:UB_obj}, it suffices to show that it is concave under the case that $g(\bs{p}), g(\bs{p}') \leq 0$,
	because $UB_\bf$ is continuous and when $g(\bs{p}) > 0$, it is affine.

	Under this case, we can write $UB_\bf(t\bs{p} + (1-t)\bs{p}') = 1 - \frac{Y}{2X}G(t)$,
	where $G(t) := \frac{n(t)}{d(t)}$, $n(t):= (t\bs{v}^\top\bs{p} + (1-t)\bs{v}^\top\bs{p}')^2$, and $d(t) := t v_\bf p_\bf + (1-t) v_\bf p_\bf'$.
	To prove the result, it now suffices to show that $G''(t) \geq 0$.

	Applying the quotient rule twice, the denominator of $G''(t)$ is $d^4(t) > 0$.
	The sign of $G''(t)$ is then the sign of its numerator, which is the sign of the expression
	\begin{equation}\label{eq:G_numerator}
		n''(t)d^2(t) - 2n'(t)d(t)d'(t) + 2(d'(t))^2n(t).
	\end{equation}
	For shorthand, let us write for every $i\in\BFs$, 
	$a_i = v_i(p_i-p_i')$ and
	$b_i = tv_ip_i+(1-t)v_ip_i'$.
	Let us also write $\gamma_\bf = tv_\bf p_\bf + (1-t)v_\bf p_\bf'$
	and $\beta_\bf = v_\bf(p_\bf - p_\bf')$.
	Then, \eqref{eq:G_numerator} can be written as 
	\begin{equation}
		2\left(\gamma_\bf \sum_{i\in\BFs} a_i - \beta_\bf \sum_{i\in\BFs} b_i \right)^2 \geq 0.
	\end{equation}
\end{proof}

From Lemma~\ref{lem:UB_concave}, we can deduce that the optimizer of \eqref{eq:OPT-Y} belongs to a finite collection of points.

\begin{lemma}\label{lem:UB_polyhedral}
	For some $\mcs\subseteq\BFs$, the point $\bs{q}(\mcs)$ is a solution of \eqref{eq:OPT-Y}. 
\end{lemma}
\begin{proof}
	This follows from the fact that the minimum of a concave function over convex polyhedral constraints is attained at a vertex of the constraint set.
	By Lemma \ref{lem:UB_conv_hull}, $Q_{\{\bf\}}$ is a convex polytope, 
	and by Lemma \ref{lem:UB_concave}, $\UB_\bf(\bs{p})$ on the domain $Q_{\{\bf\}}$ is concave.
	Therefore, $\arg\min_{\bs{p}\in Q_{\{\bf\}}} \UB_\bf(\bs{p})$ belongs to the set of vertices $\{\bs{q}(\mct)\}_{\mct\supseteq\{c\}}$.
	Let $\UB_\bf^*$ denote this minimum value.
	Then we have $\UB^* = \min_{\overbfs} \UB_\bf^*$.
\end{proof}
This lemma says that one could search over a finite number of values, $\{\UB(\bs{q}(\mcs))\}_{\mcs\subseteq\BFs}$, in order to solve \eqref{eq:OPT-Y}.
This search, however, is over a set of $2^\numbfs$ points, which scales exponentially with the number of contests.

The next result demonstrates that choosing the top $k$ contests is always weakly preferred over choosing any other subset of contests of size $k$.

\begin{lemma}\label{lem:UB_topk}
    For any subset $\mcs\subseteq \BFs$ of size $k$, i.e. $|\mcs| = k$, it holds that
    \begin{equation}
        \UB(\bs{q}([k]) \leq \UB(\bs{q}(\mcs)).
    \end{equation}
\end{lemma}

\begin{proof}
    From Lemma \ref{lem:UB_obj_equalized}, we evaluate the inequality $\UB(\bs{q}([k]) \leq \UB(\bs{q}(\mcs))$ as
    \be
        \ba
            \frac{\prod_{j\in[k]} v_j}{\sum_{\ell \in [k]} \left( \prod_{j\in [k]\setminus \ell} v_j \right)} \geq \frac{\prod_{j\in \mcs} v_j}{\sum_{\ell \in \mcs} \left( \prod_{j\in \mcs\setminus \ell} v_j \right)}.
        \ea
    \ee
    Let $\mcal{I} = [k] \cap \mcs$, and for shorthand, denote $v_{\mcal{I}} := \prod_{j\in\mcal{I}} v_j$, $v_{\mcs\setminus\mcal{I}} := \prod_{j\in\mcs\setminus\mcal{I}} v_j$, and $v_{[k]\setminus\mcal{I}} := \prod_{j\in[k]\setminus\mcal{I}} v_j$. 
	From the above inequality, we then have $\sum_{\ell \in \mcs} v_{[k]\setminus\mcal{I}} \cdot v_{\mcs\setminus \ell} \geq \sum_{\ell \in [k]} v_{\mcs\setminus\mcal{I}}\cdot v_{[k]\setminus \ell}$.
	We rewrite this relation by separating the sums,
	\begin{equation}
		\begin{aligned}
			&\sum_{\ell \in \mcal{I}} v_{[k]\setminus\mcal{I}} \cdot v_{\mcs\setminus \ell} + \sum_{\ell \in \mcs\setminus\mcal{I}} v_{[k]\setminus\mcal{I}} \cdot v_{\mcs\setminus \ell} \\
			&\quad\quad\quad \geq \sum_{\ell \in \mcal{I}} v_{\mcs\setminus\mcal{I}} \cdot v_{[k]\setminus \ell} + \sum_{\ell \in [k]\setminus\mcal{I}} v_{\mcs\setminus\mcal{I}} \cdot v_{[k]\setminus \ell}
		\end{aligned}.
	\end{equation}
    Take any $\ell \in \mcal{I}$. We verify that $v_{[k]\setminus\mcal{I}} \cdot v_{\mcs\setminus \ell} = v_{\mcs\setminus\mcal{I}}\cdot v_{[k]\setminus \ell}$. Indeed, we see that
    \be
        \frac{v_{\mcs\setminus\ell}}{v_{\mcs\setminus\mcal{I}}} = \frac{v_{[k]\setminus\ell}}{v_{[k]\setminus\mcal{I}}} = v_{\mcal{I}\setminus\ell}.
    \ee
    Consequently, the inequality becomes equivalent to:
    \be
        \sum_{\ell \in \mcs\setminus\mcal{I}} v_{[k]\setminus\mcal{I}} \cdot v_{\mcs\setminus \ell} \geq \sum_{\ell \in [k]\setminus\mcal{I}} v_{\mcs\setminus\mcal{I}}\cdot v_{[k]\setminus \ell}.
    \ee
    Now, take any $\ell \in \mcs\setminus\mcal{I}$ and any $\ell' \in [k]\setminus\mcal{I}$. We verify that $v_{[k]\setminus\mcal{I}} \cdot v_{\mcs\setminus \ell} \geq v_{\mcs\setminus\mcal{I}}\cdot v_{[k]\setminus \ell'}$. Indeed, this condition is equivalent to:
    \be
        \frac{v_{\mcs\setminus \ell}}{v_{\mcs\setminus\mcal{I}}} = \frac{v_{\mcal{I}}}{v_\ell} \geq \frac{v_{[k]\setminus \ell'}}{v_{[k]\setminus\mcal{I}}} = \frac{v_{\mcal{I}}}{v_{\ell'}}.
    \ee
    This is satisfied because for any $\ell \in \mcs\setminus\mcal{I}$ and $\ell' \in [k]\setminus\mcal{I}$, we have $\ell > \ell'$, and thus $v_\ell \leq v_{\ell'}$, which is precisely the above condition.
\end{proof}

Lemmas \ref{lem:UB_polyhedral} and \ref{lem:UB_topk} together assert that the solution to \eqref{eq:OPT-Y} can be found by searching over the collection of $\numbfs$ points $\{\bs{q}([k])\}_{k=1}^\numbfs$.
This establishes the proof of Lemma~\ref{lem:UB_analytical}.

\end{document}